%% file: Arxiv_full.tex
\documentclass[11pt]{article}
\usepackage{amsmath,amsfonts,amsthm,amssymb}
\usepackage{algorithm}
\usepackage{graphics,color,url}
\usepackage{graphicx}
\usepackage{epsfig}
\usepackage{fullpage}
\usepackage{color}
\usepackage{xfrac}

\DeclareMathOperator*{\argmax}{\ensuremath{\mathop{\rm argmax}}}

\newtheorem{theorem}{Theorem}[section]
\newtheorem{lemma}{Lemma}[section]
\newtheorem{claim}{Claim}[section]

\newcommand{\Exp}{{\bf E}}
\newcommand{\opt}{{\tt opt}}

\newcommand{\R}{{\mathbb{R}}}

\newcommand{\pr}{{\mathbf{Pr}}}
\newcommand{\ex}{{\mathbf{E}}}

\newcommand{\calD}{{\cal D}}

\newcommand{\calM}{{\cal M}}

\newcommand{\I}{{\mathcal I}}
\newcommand{\paymentsharing}{{\sc SA-random-sample}}
\newcommand{\submain}{{\sc SA-mechanism-main-2}}
\newcommand{\submax}{{\sc SA-alg-max}}

\newcommand{\XOSsample}{{\sc XOS-random-sample}}
\newcommand{\XOSmain}{{\sc XOS-mechanism-main}}
\newcommand{\AddM}{{\sc Additive-mechanism}}

\newcommand{\SubMainm}{{\sc SA-mechanism-main}}
\newcommand{\nv}{{\widetilde{v}}}

\title{Budget Feasible Mechanism Design: From Prior-Free to Bayesian}

\author{
Xiaohui Bei\thanks{Tsinghua University, China. Email: {\tt beixiaohui@gmail.com}.}\\
\and
Ning Chen\thanks{Division of Mathematical Sciences, School of Physical and Mathematical Sciences, Nanyang Technological University, Singapore. Email: {\tt ningc@ntu.edu.sg, ngravin@pmail.ntu.edu.sg}. } \\
\and
Nick Gravin$^{\dag}$\\
\and Pinyan Lu\thanks{Microsoft Research Asia. Email: {\tt pinyanl@microsoft.com}.}
}\date{}

\begin{document}

\maketitle
\begin{abstract}
Budget feasible mechanism design studies procurement combinatorial auctions
in which the sellers have private costs to produce items, and the buyer
(auctioneer) aims to maximize a social valuation function on subsets of items,
under the budget constraint on the total payment. One of the most important
questions in the field is ``which valuation domains admit truthful budget
feasible mechanisms with `small' approximations (compared to the social optimum)?''
Singer~\cite{PS10} showed that additive and submodular functions have a constant
approximation mechanism. Recently, Dobzinski, Papadimitriou, and Singer~\cite{DPS11}
gave an $O(\log^2n)$ approximation mechanism for subadditive functions;
further, they remarked that: {\em ``A fundamental question is whether, regardless
of computational constraints, a constant-factor budget feasible
mechanism exists for subadditive functions."}

In this paper, we address this question from two viewpoints:
prior-free worst case analysis and Bayesian analysis, which are two standard
approaches from computer science and economics, respectively.
\begin{itemize}
\item For the prior-free framework, we use a linear program (LP) that describes
the fractional cover of the valuation function; the LP is also
connected to the concept of approximate core in cooperative game theory.
We provide a mechanism for subadditive functions whose approximation is $O(\I)$,
via the worst case integrality gap $\I$ of this LP.
This implies an $O(\log n)$-approximation for subadditive valuations,
$O(1)$-approximation for XOS valuations, as well as for valuations having a
constant integrality gap. XOS valuations are an important class of functions and
lie between the submodular and the subadditive classes of valuations.
We further give another polynomial time $O(\frac{\log n}{\log\log n})$ sub-logarithmic
approximation mechanism for subadditive functions.\\
Both of our mechanisms improve the best known approximation ratio $O(\log^2 n)$.

\item For the Bayesian framework, we provide a constant approximation mechanism for all subadditive functions,
using the above prior-free mechanism for XOS valuations as a subroutine.
Our mechanism allows correlations in the distribution of private information and is universally truthful.
\end{itemize}
\end{abstract}

\newpage

\section{Introduction}

Consider a procurement combinatorial auction problem where there is a buyer who wants to purchase resources from a set of agents $A$. Each $i\in A$ is able to supply a resource at an incurred cost $c(i)$. The buyer has a sharp budget $B$ that gives an upper bound on the compensation that is distributed among agents, and a function $v(\cdot)$ describing the valuation that the buyer obtains for each subset of $A$. This defines a natural optimization problem: find a subset $S\subseteq A$ that maximizes $v(S)$ subject to $\sum_{i\in S}c(i)\le B$. The budgeted optimization problem has been considered in a variety of domains with respect to different valuation functions, e.g., additive (a.k.a.~knapsack), and submodular~\cite{Svi04}.

Agents, as self-interested entities, may want to get as many subsidies as possible. In particular,
an agent can hide his true incurred cost $c(i)$ (which is known only to himself) and claim `any' amount $b(i)$ instead.
We adopt the approach of mechanism design to manage self-interested, but strategic, behaviors of the agents:
Given submitted bids $b(i)$ from all agents, a {\em mechanism} decides a winning set $S\subseteq A$ and a payment $p(i)$ to each winner $i\in S$.
A mechanism is called {\em truthful} (a.k.a.~incentive compatible) if for every agent it is a dominant strategy to bid his true cost\footnote{\small The focus of our work is to consider strategic behaviors of the agents rather than the buyer. We thus assume that the information related to the buyer, i.e., budget $B$ and valuation function $v(\cdot)$, is public knowledge.}, i.e., $b(i)=c(i)$.
Truthfulness is one of the central solution concepts in mechanism design. It ensures that every participant will behave truthfully to his best interest.

Our mechanism design problem has an important and practical ingredient: the budget, i.e., the total payment of a mechanism
should be upper bounded by $B$. The budget constraint introduces a new dimension to mechanism design and restricts the
space of truthful mechanisms. For example, in single-parameter domains where the private information of every individual
is a single value (which is the case in our model), a monotone allocation rule with associated threshold payments provides a sufficient and necessary
condition for truthfulness~\cite{myerson}. However, it may not necessarily generate a budget feasible solution.
Thus, a number of well known truthful designs (e.g., the seminal VCG mechanism~\cite{vickrey,clarke,groves}) do not apply, and new ideas have to be developed.

Another significant challenge due to the budget constraint is that, unlike the VCG mechanism which always generates a socially optimal solution, we cannot hope to have a solution that is both socially optimal and budget feasible even if we are given unlimited computational power. Indeed, in a simple setting like path procurement with $0$ or $1$ valuation~\cite{PS10},
any budget feasible mechanism may have an arbitrarily bad solution. Therefore, the question that one
may ask is ``under which valuation domains do there exist truthful budget feasible mechanisms that admit `small' approximations
(compared to the socially optimal solution)?''

The answer to this question crucially depends on the properties and classifications of the valuation function under consideration.
In particular, given the following function hierarchy~\cite{LLN01}:
\begin{eqnarray*}
\text{additive \ \ } \subset \textup{ \ \ gross substitutes \ \ $\subset$ \ \ submodular \ \ } \subset \text{ \ \ XOS \ \ $\subset$ \ \ subadditive,}
\end{eqnarray*}
which one admits a positive answer?
Singer~\cite{PS10} initiated the study of approximate budget feasible mechanism design and gave constant approximation mechanisms for additive and submodular functions.
In subsequent work, Dobzinski, Papadimitriou, and Singer~\cite{DPS11} considered subadditive functions and showed an $O(\log^2n)$ approximation.
Further, it was remarked in~\cite{DPS11} that:
\begin{quote}
{\em ``A fundamental question is whether, regardless of computational constraints, a constant-factor budget feasible mechanism exists for subadditive functions."}
\flushright --- Dobzinski, Papadimitriou, Singer
\end{quote}
In the present paper we attempt to answer this question.

\subsection{Our Results and Techniques}

We address this question from two viewpoints: prior-free worst case analysis and Bayesian analysis.
The former is the standard framework used in computer science: in our model, the private cost $c(i)$ of every
agent is assumed to be arbitrary without any prior knowledge. All previous research on budget feasible mechanism design,
e.g. \cite{PS10,CGL11,DPS11,GR} falls into this framework. The latter Bayesian analysis~\cite{myerson}
is a classic economic approach that assumes the private information of the agents is drawn from a given prior-known distribution.
Bayesian mechanism design has received a lot of attention in the computer science community in recent years, see, e.g.,
~\cite{HR08,HR09,HL10,BGGM10,CHMS10,DRY10,CMS10,HKM11,BH11,CMM11,DHK11}.

\paragraph{Prior-free mechanism design}

Consider the following linear program (LP), where $\alpha(\cdot)$'s are variables.
\begin{eqnarray*}
&\min& \sum_{S\subseteq A} \alpha(S)\cdot v(S)\\
&s.t.& \alpha(S)\ge 0,\quad\quad \forall\ S\subseteq A \\
&& \sum_{S:\ i\in S} \alpha(S) \ge 1,\quad \forall\ i\in A
\end{eqnarray*}

Constraints of this LP describe a fractional set cover of $A$, where each set $S$
receives weight $\alpha(S)$ and we require that all agents in $A$ are covered.
%
%
An important observation about this LP is that for any monotone subadditive
function $v(\cdot)$, the value of the optimal integral solution is precisely $v(A)$.

The above LP has a strong connection to the core of cost sharing games
(considering $v(\cdot)$ instead as a cost function), which is a
central notion in cooperative game theory~\cite{agt-book}. Roughly
speaking, the core of a game is a stable cooperation among all
agents to share $v(A)$ where no subset of agents can benefit by
breaking away from the grand coalition. It is well known that the
cores of many cost sharing games are empty. This motivates the notion
of $\alpha$-approximate core, which requires all the agents to share
only an $\alpha$-fraction of $v(A)$. The classic Bondareva-Shapley
Theorem~\cite{bondareva63,shapley67} says that for subadditive
functions, the largest value $\alpha$ for which the
$\alpha$-approximate core is nonempty is equal to the integrality
gap of the LP. Further, the integrality gap of the LP equals one (i.e.,
$v(A)$ is also an optimal fractional solution) if and only if the
valuation function is XOS, which is also equivalent to the
non-emptiness of the core.

Given an instance of our problem with an agent set $A$, we may
consider the above LP and its integrality gap for every subinstance
$A'\subseteq A$. We denote $\I$ as the largest integrality gap among all subinstances $A'\subseteq A$.
In other words, the gap $\I$ characterizes the worst case scenario between the optimal integral and
fractional solutions of the problem. We have the following result.

\medskip
\noindent \textbf{Theorem 1.} \textit{There is a budget feasible
truthful mechanism for subadditive functions with approximation
ratio $O(\I)$. In particular, for XOS functions, the mechanism has a
constant approximation ratio.}
\medskip

Our results show an interesting connection between the budget feasible
mechanism design and the integrality gap of the above linear program,
as well as the existence of an $\alpha$-approximate core.
Note that the tight bound of the integrality gap is known to be $\Theta(\log n)$~\cite{Dob07,BR11};
thus, our mechanism in the worst case has an approximation ratio of $O(\log n)$.
(But for some special functions whose integrality gaps are bounded by constants, e.g., facility location~\cite{agt-book}, our mechanism gives a constant approximation.)
Further, the mechanism may have exponential running time, though for some special XOS functions like matching
it can be implemented in polynomial time.
To remedy these issues, we further give a polynomial time sub-logarithmic approximation mechanism.
Both of our mechanisms improve the best known approximation ratio $O(\log^2 n)$ of~\cite{DPS11}.

\medskip
\noindent \textbf{Theorem 2.} \textit{There is a polynomial time budget feasible truthful mechanism for
subadditive functions with an approximation ratio $O(\frac{\log n}{\log\log n})$, where $n$ is the
number of agents.}


\paragraph{Bayesian mechanism design}
As a standard game theoretic model for incomplete information, Bayesian mechanism design assumes that agents'
private information (i.e., $c(i)$ in our model) is drawn from a known distribution.
In contrast to prior-free worst case analysis, if we have prior knowledge of the distribution,
we can obtain more positive results in the form of constant approximation truthful mechanisms.
In the Bayesian setting, we are able to answer the above question posed in~\cite{DPS11} affirmatively.

\medskip
\noindent \textbf{Theorem 3.} \textit{There is a constant approximation budget feasible truthful mechanism
for subadditive functions for any distribution under a mild assumption\footnote{\small Technically,
we require that the distribution has integrable marginal densities for any subset of variables (e.g.,
jointly independent distributions trivially satisfy this condition). The formal definition is referred to Section~\ref{section-bayesian}.}.}
\medskip

It should be noted that our result does not completely rely on Bayesian analysis in the following aspects.

\begin{itemize}
\item Truthfulness. In most of the previous works in Bayesian mechanism design regarding social welfare maximization, e.g.,~\cite{HL10,BH11,HKM11,CMM11}, the considered solution concept is Bayesian truthfulness,
i.e., truth-telling is in expectation an equilibrium strategy when other agents' profiles are drawn from the known distribution.
Our mechanism guarantees {\em universal truthfulness}, meaning that truth-telling is a dominant strategy of each agent
for any coin flips of the mechanism and any instance of the costs. Thus universal truthfulness is a stronger solution concept than
Bayesian truthfulness. Universal truthfulness in Bayesian mechanism design has also been used in, e.g.~\cite{CHMS10}, but their focus was on profit maximization.

\item Distribution. Regarding prior knowledge of the distribution,
  most of the previous related works consider independent distributions, e.g.,~\cite{HR08,HL10,BH11,HKM11}.
  Our mechanism applies to general distributions that allow
  correlations on costs. Dependency on private information is a natural phenomenon arising in practice
  and it has been considered for, e.g., auctions~\cite{Mil04}. In our model where costs are private information,
  correlations appear to be very common. For example, if the price of crude oil goes up, the costs of producing the
  items for {\em all} agents may go up as well.
\end{itemize}

\vspace{-0.1in}
\paragraph{Techniques}
In the design of budget feasible mechanisms, the major approach used in previous works~\cite{PS10,CGL11,DPS11} is based on a simple idea of adding agents one by one greedily
and carefully ensuring that the budget constraint is not violated. Our mechanisms, from a high
level structural point of view, use another simple, but powerful, approach: random sampling.
We add agents into a test set $T$ with probability half for each agent and compute an
(approximately) optimal budget feasible solution on $T$.
We use the agents in $T$ only for the purpose of `evaluation' and none of
them will be a winner. The computed solution on $T$ gives a close estimate
of the optimal solution for the whole set with a high probability. We
then compute a real winning set from the remaining agents using the evaluation
from $T$ as a threshold.

In the Bayesian setting, random sampling is often deemed to be unnecessary, because, when we have knowledge of the distribution, it
is tempting to use a `prior sampling' approach to generate random virtual instances and
based on them to compute a threshold. While
this works well when the private cost $c(i)$ of every agent is drawn independently,
interestingly (and surprisingly), it fails when costs $c(i)$'s are correlated in the distribution.
We therefore still use random sampling to compute a threshold based on the sampled test set;
the collected information from random sampling correctly reflects the structure of the private costs (with a high probability)
even for correlated distributions. To derive a constant approximation budget feasible mechanism for subadditive functions,
we first generate a cost vector sampled from the known distribution and use it as a guidance for the payments to the winners.
Then we adopt our (constant approximation) mechanism for prior-free XOS functions by feeding to this mechanism another valuation $\nv(\cdot)$,
which we define as the solution of the above LP computed for the various subsets of $A$.


Random sampling appears to be a powerful approach and has been used successfully in other domains
of mechanism design, e.g., digital goods auctions~\cite{GHK}, secretary problem~\cite{BIK,BDG}, social welfare maximization~\cite{Dob07},
and mechanism design without money~\cite{CGL}. It is intriguing to find applications of random
sampling in other mechanism design problems.


\subsection{Related Work}

Our work falls into the field of algorithmic mechanism design,
which is a fascinating area initiated by the seminal work of Nisan and Ronen~\cite{NR99}.
There are many mechanism design models; see, e.g.,~\cite{agt-book}, for a survey.

As mentioned earlier, the study of approximate mechanism design with
a budget constraint was originated by Singer~\cite{PS10} and constant
approximation mechanisms were given for additive and submodular
functions. The approximation ratios were later improved
in~\cite{CGL11}. Dobzinski, Papadimitriou, and Singer~\cite{DPS11}
considered subadditive functions and showed an $O(\log^2n)$
approximation mechanism.
Ghosh and Roth~\cite{GR} considered a budget feasible
mechanism design model for selling privacy where there are
externalities for each agent's cost. All these models considered prior-free worst case analysis.

For Bayesian mechanism design, Hartline and Lucier~\cite{HL10} first proposed a Bayesian reduction in single-parameter settings that
converts any approximation algorithm to a Bayesian truthful mechanism that approximately preserves social welfare.
The blackbox reduction results were later improved to multi-parameter settings in~\cite{BH11} and~\cite{HKM11} independently.
Chawla et al.~\cite{CMM11} considered budget-constrained agents and gave Bayesian truthful mechanisms in various settings.
A number of other Bayesian mechanism design works considered profit maximization, e.g.,~\cite{HR08,BGGM10,CMS10,DRY10,CHMS10,DHK11}.
Ours is the first to consider Bayesian analysis in budget feasible mechanisms with a focus on the valuation (social welfare) maximization.


\section{Preliminaries}\label{sec:preliminaries}

In a marketplace, there are $n$ agents (or items), denoted by $A$. Each agent $i\in A$ has a privately known incurred {\em cost} $c(i)\ge 0$.
We denote by $c=(c(i))_{i\in A}$ the cost vector of the agents.
For any given subset $S\subseteq A$ of agents, there is a publicly known valuation $v(S)$, meaning the social welfare derived from $S$. We assume that $v(\emptyset)=0$ and the valuation function is monotone, i.e., $v(S)\le v(T)$ for any $S\subset T\subseteq A$. A centralized authority wants to pick a subset of agents with maximum possible valuation given a sharp budget $B$ to cover their
incurred costs, i.e., $\max_{_{S\subseteq A}}v(S)$ given $c(S) = \sum_{i\in S}c(i) \le B$.
We denote the optimal solution of this optimization question by $\opt(A)$ (or $\opt(c)$) and its valuation by $v(\opt(A))$.

We will consider XOS and subadditive functions in the paper; both are rather general classes and contain a number
of well studied functions as special cases, e.g., additive, gross substitutes, and submodular.
\begin{itemize}
\item Subadditive (a.k.a.~complement free): $v(S)+v(T)\ge v(S\cup T)$ for any $S,T\subseteq A$.
\item XOS (a.k.a.~fractionally subadditive): there is a set of linear functions $f_1,\ldots,f_m$ such that $$v(S)=\max\big\{f_1(S),f_2(S),\ldots,f_m(S)\big\}$$ for any $S\subseteq A$. Note that the number of functions $m$ can be exponential in $n=|A|$.

   An equivalent definition~\cite{feige09} is $v(S)\le \sum_{T\subseteq A}x(T)\cdot v(T)$ whenever $\sum_{T\subseteq A:~i\in T}x(T)\ge 1$ for any $i\in S$, where $0\le x(T)\le 1$. That is, if every element in $S$ is fractionally covered, then the sum of the values of all subsets weighted by the corresponding fractions is at least as large as $v(S)$.
\end{itemize}

Note that the representation of a valuation function usually requires exponential size in $n$.
Thus, we assume that we are given an access to a {\em demand oracle}, which, for any given price vector $p(1),\ldots,p(n)$,
returns us a subset $T\in \argmax_{_{S\subseteq A}}\left(v(S)-\sum_{i\in S}p(i)\right)$; every such query is assumed to take unit time.
The demand oracle is used in~\cite{DPS11} as well, and it was shown that a weaker value query oracle is not sufficient~\cite{PS10}.

Agents, as self-interested entities, have their own objective as well; each agent $i$ may not tell his true
privately known cost $c(i)$, but, instead, submit a {\em bid} $b(i)$ strategically.
We use mechanism design and its solution concept truthfulness to manage strategic behaviors of the agents.
Upon receiving $b(i)$ from each agent, a mechanism decides an {\em allocation} $S\subseteq A$ of the winners and a {\em
payment} $p(i)$ to each $i\in A$. We assume that the mechanism has no positive transfer (i.e., $p(i)=0$ if $i\notin S$)
and is individually rational (i.e., $p(i)\ge b(i)$ if $i\in S$).

In a mechanism, agents bid strategically to maximize their utilities, which is $p(i)-c(i)$ if $i$ is a winner and $0$
otherwise. We say a mechanism is {\em truthful} if it is of the best interest for each agent to report his true cost,
i.e., $b(i)=c(i)$. For randomized mechanisms, we consider universal truthfulness in this paper: a randomized mechanism is
called {\em universally truthful} if it takes a distribution over deterministic truthful mechanisms.

Note that our model is in a single parameter domain, as each agent has only one private parameter which is his cost. Thus, by the well known characterization of single parameter truthful mechanisms~\cite{myerson}, designing a monotone allocation, plus the corresponding threshold payment rule, is sufficient to derive a truthful mechanism. We therefore do not specify the payments to the winners in our mechanisms explicitly.

A mechanism is said to be {\em budget feasible} if its total payment is within the budget constraint, i.e., $\sum_i p(i)
\le B$. Our goal in the present paper is to design truthful and budget feasible mechanisms for XOS and subadditive functions in two frameworks: prior-free and Bayesian.

We first establish the following technical lemma, which is useful in the analysis of our mechanisms.

\begin{lemma}\label{lem:probability}
Consider any subadditive function $v(\cdot)$. For any given subset $S\subseteq A$ and a positive integer $k$,
we assume that $v(S)\ge k\cdot v(i)$ for any $i\in S$. Further, suppose
that $S$ is divided uniformly at random into two groups $T_1$ and
$T_2$. Then, with probability of at least $\frac{1}{2}$, we have
$v(T_1)\ge \frac{k-1}{4k}v(S)$ and $v(T_2)\ge \frac{k-1}{4k}v(S)$.
\end{lemma}

\section{Prior-Free Mechanism Design}

In this section we consider designing budget feasible mechanisms for XOS and subadditive functions in the prior-free setting.
That is, the mechanism designer has no prior knowledge of the private information $c(i)$ of every agent, which can be an arbitrary cost,
and the performance of a designed mechanism is analyzed in the worst case framework.
That is, we evaluate a mechanism according to its {\em approximation ratio}, which is defined as $\max_{c} \frac{v(\opt(c))}{\mathcal{M}(c)}$, where $\mathcal{M}(c)$ is the (expected) value of a mechanism $\mathcal{M}$ on instance $c=(c(i))_{i\in A}$ and $v(\opt(c))$ is its optimal value.
(We assume without loss of generality that $c(i)\le B$ for any $i\in A$, since such an agent will never win in any budget feasible truthful mechanism.)

\subsection{Constant Approximation for XOS}\label{sec:xos}

We will first consider XOS functions. Given an XOS function $v(\cdot)$, by its definition,
we assume that $$v(S)=\max\left\{f_1(S),f_2(S),\ldots,f_m(S)\right\}$$ for any $S\subseteq A$, where each $f_j(\cdot)$ is a nonnegative additive function,
i.e., $f_j(S)=\sum_{i\in S}f_j(i)$.

In our mechanism, we use a random mechanism \AddM\ for additive valuation functions
as an auxiliary procedure, where \AddM\ is a universally truthful mechanism
and has an approximation factor of at most $3$ (see, e.g., Theorem~B.2, \cite{CGL11}).

\begin{center}
\small{}\tt{} \fbox{
\parbox{6.4in}
{\hspace{0.05in} \\
\underline{\XOSsample}
\begin{enumerate}
\item Pick each item independently at random with probability $\frac{1}{2}$ into group $T$.
\item Compute an optimal solution $\opt(T)$ for items in $T$ given budget $B$.
\item Set a threshold $t=\frac{v(\opt(T))}{8B}.$
\item Consider items in $A\setminus T$ and find a set  $S^*\in \argmax\limits_{S\subseteq A\setminus T}\big\{v(S)-t\cdot c(S)\big\}.$
\item Find an additive function $f$ with $f(S^*)=v(S^*)$ in the XOS representation of $v(\cdot)$.
\item Run \AddM\ for function $f(\cdot)$ with respect to set $S^*$ and budget $B$.
\item Output the result of \AddM.
\end{enumerate}
} }
\end{center}

In the above mechanism, we first sample in expectation half items to form a testing group $T$, and compute an optimal solution for $T$ given budget constraint $B$.
By Lemma~\ref{lem:probability}, we know that $v(\opt(A)) \ge v(\opt(T))\ge \frac{k-1}{4k}v(\opt(A))$ and $v(\opt(A\setminus T))\ge \frac{k-1}{4k}v(\opt(A))$ with a probability of at least $\frac{1}{2}$.
That is, we are able to learn the rough value of the optimal solution by random sampling, and still keep a nearly optimal solution formed with the remaining items.
We then use the information from random sampling to compute a proper threshold $t$ for the rest of items.
Specifically, we find a subset $S^*\subseteq A\setminus T$ with the largest difference between its value and cost, multiplied by the threshold $t$
(in the computation of $S^*$, if there are multiple choices, we break ties by any given fixed order).
Finally, we use the property of XOS functions to find a linear representation of $v(S^*)$ and run a truthful mechanism for linear functions with respect to $S^*$.

The mechanism is designated for XOS functions; it is also used crucially as an auxiliary procedure for the more general subadditive functions in the subsequent sections.
Note that the runtime of the mechanism for general XOS functions is exponential\footnote{\small In fact, in the second step of the mechanism, we can use any constant approximation solution (e.g., algorithm \submax\ established in Section~\ref{section-sub-log}), which suffices for our purpose.
Further, Step~(4) can be done by simply asking a demand query. Hence, the mechanism can be implemented in polynomial time,
if we have access to an oracle that, for any subset $X$ of items, gives a linear function $f$ with $f(X)=v(X)$ and $f(S)\le v(S)$ for each $S\subseteq X$. For some classic XOS problems like matching (the value of a subset of edges is the size of the largest matching induced by them), our mechanism can be implemented in polynomial time.}.
\normalsize

Note that in Step~(4), the function $v(S)-t\cdot c(S)$ that we maximize
is simply the Lagrangian function
$$v(S)-x\cdot c(S) + x\cdot B$$ (note that $x\cdot B$ is a fixed constant) of
the original optimization problem $\max_{_{S}} v(S)$ subject to $c(S)\le B$.
While we do not know the actual value of the variable $x$ in the Lagrangian, a
carefully chosen parameter $t$ in the sampling step with a high probability
ensures that $\max_{_{S}}\big\{v(S)-t\cdot c(S)+t\cdot B\big\}$ gives a constant approximation of the optimum $\max_{_{S}}\big\{v(S)-x\cdot c(S) + x\cdot B\big\}$ of the Lagrangian,
which is precisely the targeted value $v(\opt(A))$.

The linearity of the Lagrangian, together with the subadditivity of the valuations,
is important in order to derive the following properties.
(The threshold $t$, subset $S^*$, and additive function $f$ are defined in the \XOSsample.)

\begin{claim}\label{lem:subset}
For any $S \subseteq S^*$, $f(S) - t\cdot c(S) \geq 0$.
\end{claim}
\begin{proof}
Suppose by a contradiction that there exists a subset $S \subseteq S^*$ such that $f(S) - t\cdot c(S) < 0$. Let $S' = S^* \setminus S$.
Since $f$ is an additive function, we have $c(S') + c(S) = c(S^*)$ and $f(S') + f(S) = f(S' \cup S) = f(S^*) = v(S^*)$.
Thus,
\begin{eqnarray*}
  v(S') - t\cdot c(S') & \geq & f(S') - t\cdot c(S') \\
  & = & v(S^*) - t\cdot c(S^*) - \big(f(S) - t\cdot c(S)\big) \\
  & > & v(S^*) - t\cdot c(S^*),
\end{eqnarray*}
which contradicts the definition of $S^*$.
\end{proof}

The following claim says that any item in $S^*$ cannot manipulate the selection of the set $S^*$ if bidding a smaller cost. This fact is critical for the monotonicity, and thus, the truthfulness of the mechanism.

\begin{claim} \label{lem:SameS}
If any item $j\in S^*$ reports a smaller cost $b(j) < c(j)$, then set $S^*$ remains the same.
\end{claim}
\begin{proof}
Let $b$ be the bid vector where $j$ reports $b(j)$ and others remain unchanged. First we notice that for any set $S$ with $j \in S$,
$\big(v(S) - t\cdot b(S)\big) - \big(v(S) - t\cdot c(S)\big) = t\big(c(j) - b(j)\big)$ is a fixed positive value.
Hence,
\begin{eqnarray*}
    v(S^*) - t\cdot b(S^*) & = & v(S^*) - t\cdot c(S^*) + t\big(c(j) - b(j)\big)\\
    & \geq & v(S) - t\cdot c(S) + t\big(c(j) - b(j)\big) \\
    & = & v(S) - t\cdot b(S).
\end{eqnarray*}
Further, for any set $S$ with $j \notin S$, we have
\begin{eqnarray*}
    v(S^*) - t\cdot b(S^*) & > & v(S^*) - t\cdot c(S^*) \\
    & \geq & v(S) - t\cdot c(S) \\
    & = & v(S) - t\cdot b(S).
\end{eqnarray*}
Therefore, we conclude that $S^* = \argmax\limits_{S\subseteq A\setminus T}\big(v(S)-t\cdot b(S)\big)$;
and by the fixed tie-breaking rule, $S^*$ is selected as well.
\end{proof}

Our main mechanism for XOS functions is simply a uniform
distribution of the mechanism \XOSsample\ and one that always picks an item from $\argmax_i v(i)$.

\begin{center}
\small{}\tt{} \fbox{
\parbox{6.4in}{\hspace{0.05in} \\
\underline{\XOSmain}
\begin{itemize}
\item With half probability, run \XOSsample.
\item With half probability, pick a most-valuable item as the only winner and pay him $B$.
\end{itemize}
} }
\end{center}

\begin{theorem}\label{theorem-XOS-main}
The mechanism \XOSmain\ is budget feasible and truthful, and provides a constant approximation ratio for XOS valuation functions.
\end{theorem}

In the remaining of this section, we will prove the theorem. It follows from the following three lemmas.

\begin{lemma}\label{lemma-XOS-truthful}
\XOSmain\ is universally truthful.
\end{lemma}

Our mechanism, at a high level point of view, has a similar flavor to the mechanism composition introduced in~\cite{AH06}.
In particular, we may consider Steps~(1-4) as one mechanism of choosing candidate winners and Steps~(5-7) as the other restricted on the survived agents; then the whole mechanism is a composition of the two.
It was shown in~\cite{AH06} that if the first mechanism is composable
(i.e., truthful plus the property that any winner cannot manipulate the winner set without losing) and the
second one is truthful, then the composite mechanism is truthful.
In our mechanism \XOSmain, composability of Steps~(1-4) follows from Claim~\ref{lem:SameS} and truthfulness of Steps~(5-7) is by the property of \AddM. Therefore, the mechanism is truthful.

\begin{lemma}
\XOSmain\ is budget feasible.
\end{lemma}

In the mechanism \XOSsample, the payment to each winner is the maximum amount that the agent can bid and still win.
This amount is the minimum of the threshold bids in each of the intermediate steps (e.g., Step~(4) and~(6)). In particular,
the payment is upper bounded by the threshold of the mechanism \AddM\ in Step~(6). As \AddM\ is budget feasible~\cite{CGL11},
our mechanism \XOSsample\ is budget feasible as well. Finally, picking the largest item and paying it the whole budget are clearly a budget feasible mechanism.


\begin{lemma}
\XOSmain\ has a constant approximation ratio.
\end{lemma}

\begin{proof}
Let $\opt=\opt(A)$ denote the optimal winning set given budget $B$, and let $k = \min_{i\in \opt}{\frac{v(\opt)}{v(i)}}$.
Thus  $v(\opt) \geq k\cdot v(i)$ for each $i \in \opt$. By Lemma~\ref{lem:probability},
we have $v(\opt \cap~ T) \geq \frac{k-1}{4k}v(\opt)$ with a probability of at least $\frac{1}{2}$.
Thus, we have $v(\opt(T)) \geq v(\opt \cap T) \geq \frac{k-1}{4k}v(\opt)$ with a
probability of at least $\frac{1}{2}$ (the first inequality is because $\opt \cap T$ is a particular solution
and $\opt(T)$ is an optimal solution for set $T$ with the budget constraint).

We let $\opt^*=\opt_f(S^*)$ be the optimal solution with respect to the item set
$S^*$, additive value-function $f$ and budget $B$. In the
following we show that $f(\opt^*)$ is a good approximation of
the actual social optimum $v(\opt)$. Consider the following two
cases:
\begin{itemize}
\item $c(S^*) > B$. With such assumption, we can always find a
    subset $S' \subseteq S^*$, such that $\frac{B}{2} \leq c(S')
    \leq B$. By Claim~\ref{lem:subset}, we know $f(S') \geq
    t\cdot c(S') \geq \frac{v(\opt(T))}{8B}\cdot\frac{B}{2} \geq
    \frac{v(\opt(T))}{16}$. Then by the fact that $\opt^*$ is an optimal
    solution and $S'$ is a particular solution with budget constraint $B$, we have $f(\opt^*) \geq f(S') \geq
    \frac{v(\opt(T))}{16}\geq \frac{k-1}{64k}v(\opt)$ with a probability of
    at least $\frac{1}{2}$.

\item $c(S^*) \leq B$. Then $\opt^* = S^*$. Let $S' = \opt \backslash T$; thus, $c(S') \leq c(\opt) \leq B$.
    By Lemma~\ref{lem:probability}, we have $v(S') \geq \frac{k-1}{4k}v(\opt)$
    with a probability of at least $\frac{1}{2}$. Recall that $S^* = \argmax_{S\subseteq A\setminus T}(v(S)-t\cdot c(S))$. Then with a probability of at least $\frac{1}{2}$, we have
    \begin{eqnarray*}
      f(\opt^*) = f(S^*) & = & v(S^*) \\
      & \geq & v(S^*) - t\cdot c(S^*) \\
      & \geq & v(S') - t\cdot c(S') \\
      & \geq &  \frac{k-1}{4k}v(\opt) -\frac{v(\opt(T))}{8B}\cdot B \\
      & \geq & \frac{k-1}{4k}v(\opt) -\frac{v(\opt)}{8} \\
      & = & \frac{k-2}{8k}v(\opt).
    \end{eqnarray*}
\end{itemize}

In either case, we get
$$f(\opt^*) \geq \min\left\{\frac{k-1}{64k}v(\opt),\frac{k-2}{8k}v(\opt)\right\} \ge  \frac{k-2}{64k}v(\opt)$$
with a probability of at least $\frac{1}{2}$. At the end we output the result of
  \AddM$(f, S^*, B)$ in the last step of \XOSsample. We recall that \AddM\ has an
  approximation factor of at most 3 with respect to the optimal
  solution $f(\opt^*)$. Thus the solution given by \XOSsample\ is at least $\frac{1}{3}\cdot
  f(\opt^*) \geq \frac{1}{3} \cdot \frac{1}{2} \cdot \frac{k-2}{64k}
  v(\opt) = \frac{k-2}{384k}v(\opt)$.

  On the other hand, since $k = \min_{i\in \opt}{\frac{v(\opt)}{v(i)}}$, the solution given by
  picking the largest item satisfies $\max_i{v(i)} \geq \frac{1}{k}v(\opt)$.
  Combining the two mechanisms together, our main mechanism \XOSmain\ has a performance of at least
  \[\left(\frac{1}{2}\cdot\frac{k-2}{384k} + \frac{1}{2}\cdot\frac{1}{k}\right)v(\opt) = \frac{k+382}{768k}v(\opt) \ge \frac{1}{768}v(\opt).\]
  This completes the proof of the lemma.
\end{proof}

\subsection{Integrality-Gap Approximations for Subadditive}
\label{sec:subadditive}

Next we use our result for XOS functions to design a truthful mechanism for subadditive functions.
Let $S_1,\ldots,S_N$ be a permutation of all possible subsets of $A$, where $N=|2^{A}|$ is the size of the power set $2^A$.
We consider the following linear program for each subset $S\subseteq A$, where each subset $S_j$ is associated with a
variable $\alpha_j$.
\begin{eqnarray*}
LP(S):~~~ &\min& \sum\limits_{j=1}^N \alpha_{j}\cdot v(S_j) \hspace{0.7in} (\lozenge) \\
&s.t.& \alpha_j\ge 0,\quad 1\le j\le N \\
&& \sum\limits_{j:\ i\in S_j} \alpha_j \ge 1,\quad \forall\ i\in S
\end{eqnarray*}
In the above linear program, the minimum is taken over all possible non-negative values of $\alpha=(\alpha_1,\ldots,\alpha_N)$.
If we consider each $\alpha_j$ as the fraction of the cover by subset $S_j$, the last constraint implies that all items in $S$ are fractionally covered.
Hence, LP$(S)$ describes a linear program for the set cover of $S$.
For any subadditive function $v(\cdot)$, it can be seen that the
value of the optimal integral solution to the above LP$(S)$ is
always $v(S)$. Indeed, one has $S\subseteq \bigcup_{j:\ \alpha_j\ge
1} S_j$ and $\sum_{j}\alpha_{j}\cdot v(S_j)\ge\sum_{j:\ \alpha_j\ge
1}v(S_j)\ge v\big(\bigcup_{j:\ \alpha_j\ge 1} S_j\big)\ge v(S)$.

Let $\nv(S)$ be the value of the optimal fractional solution of
LP$(S)$, and $\I(S)=\frac{v(S)}{\nv(S)}$ be the integrality gap of
LP$(S)$. Let $\I=\max_{S\subseteq A}\I(S)$; the integrality gap $\I$
gives a worst case upper bound on the integrality gap of all
subsets. Hence, we have $\frac{v(S)}{\I}\le\nv(S)\le v(S)$ for any
$S\subseteq A$. The classic Bondareva-Shapley
Theorem~\cite{bondareva63,shapley67} says that the integrality gap
$\I(S)$ is one (i.e., $v(S)$ is also an optimal fractional solution
to the LP) if and only if $v(\cdot)$ is an XOS function.

\begin{lemma}\label{lemma-nv-XOS}
$\nv(\cdot)$ is an XOS function.
\end{lemma}
\begin{proof}
For any subset $S\subseteq A$, consider any non-negative vector $\gamma=(\gamma_1,\ldots,\gamma_N)\ge 0$ that satisfies $\sum_{\substack{j:~ i\in S_j}} \gamma_j \ge 1$ for any $i\in S$.
Then, we have
\small
\begin{eqnarray*}
\sum_{j=1}^{N}\gamma_j\cdot \nv(S_j) &=& \sum_{j=1}^{N}\gamma_j\cdot \min_{\beta_{j,\cdot}\ge 0} \left(\sum_{k=1}^{N}\beta_{j,k}\cdot v(S_k) ~\middle|~ \forall\ i\in S_j, \sum_{\substack{k:~ i\in S_k}} \beta_{j,k} \ge 1 \right) \\
&=& \min_{\beta\ge 0} \left(\sum_{j=1}^N \gamma_j \sum_{k=1}^N \beta_{j,k} \cdot v(S_k) ~\middle|~ \forall\ j, \forall\ i\in S_j, \sum_{\substack{k:~ i\in S_k}} \beta_{j,k} \ge 1 \right) \\
&=& \min_{\beta\ge 0} \left(\sum_{k=1}^N \bigg(\sum_{j=1}^N \gamma_j \beta_{j,k} \bigg)\cdot v(S_k) ~\middle|~ \forall\ j, \forall\ i\in S_j, \sum_{\substack{k:~ i\in S_k}} \beta_{j,k} \ge 1 \right) \\
&\ge& \min_{\alpha\ge 0} \left(\sum_{k=1}^{N}\alpha_k\cdot v(S_k) ~\middle|~ \forall\ i\in S, \sum_{\substack{k:~ i\in S_k}} \alpha_k \ge 1 \right) \\
&=& \nv(S)
\end{eqnarray*}
\normalsize
The inequality above follows from the fact that for any $i\in S$,
\[
\sum_{k:~i\in S_k} \sum_{j} \gamma_j \beta_{j,k} = \sum_{j}\gamma_j \sum_{k:~i\in S_k} \beta_{j,k}
\ge \sum_{j}\gamma_j  \ge \sum_{j:~i\in S_j}\gamma_j \ge 1.
\]
Hence, $\nv(\cdot)$ is fractionally subadditive, which is equivalent to XOS.
\end{proof}

We are now ready to present our mechanism for subadditive functions.

\begin{center}
\small{}\tt{} \fbox{
\parbox{6.4in}{\hspace{0.05in} \\
\underline{\SubMainm}
\begin{enumerate}
\item For each subset $S\subseteq A$, compute $\nv(S)$.
\item Run \XOSmain\ for the instance with respect to the XOS function $\nv(\cdot)$.
\item Output the result of \XOSmain.
\end{enumerate}
} }
\end{center}

\begin{theorem}\label{theorem-integrality}
The mechanism \SubMainm\ is budget feasible and truthful, and provides an approximation ratio of $O(\I)$ for subadditive functions,
where $\I$ is the largest integrality gap of LP$(S)$ for all $S\subseteq A$.
\end{theorem}
\begin{proof}
Note that the valuation $v(\cdot)$ is public knowledge and utilities of agents do not depend on $v(\cdot)$; thus computing $\nv(\cdot)$ and running \XOSmain\ with respect to $\nv(\cdot)$ do not affect truthfulness.
The claim then follows from Theorem~\ref{theorem-XOS-main} and the fact that $\frac{v(S)}{\I}\le\nv(S)\le v(S)$ for any $S\subseteq A$
(i.e., by using $\nv(\cdot)$ instead of $v(\cdot)$ we lose at most factor of $\I$ in the approximation ratio).
\end{proof}

%

In general, the approximation ratio of the mechanism can be as large as $\Theta(\log n)$~\cite{Dob07,BR11}.
But for those instances when the integrality gap of $(\lozenge)$ is bounded by a constant (e.g., facility location~\cite{agt-book}),
our mechanism gives a constant approximation.

\subsection{Sub-Logarithmic Approximations for Subadditive}\label{section-sub-log}

In this section, we give another mechanism for subadditive functions based on the ideas of
random sampling and cost sharing. In contrast to the previous section, the mechanism runs in
polynomial time and has an $o(\log n)$ approximation ratio, improving the previously
best known ratio $O(\log^2 n)$~\cite{DPS11}.
Our mechanism relies on a constant factor approximation {\em algorithm} for subadditive
function maximization under a knapsack constraint, which may have its own interest.

\subsubsection{Subadditive Maximization with Budget}\label{appendix-approx-alg}

We first give an algorithm that approximates $\max_{_{S\subseteq A}}v(S)$ given that $c(S)\le B$.
That is, we ignore for a while strategic behaviors of the agents and consider a pure optimization problem.
Dobzinski et al.~\cite{DPS11} considered the same question and gave a 4-approximation algorithm for
the unweighted case (i.e., the restriction is on the size of a selected subset).
Our algorithm extends their result to the weighted case and runs in polynomial time if we are given a demand oracle\footnote{\small Independent to our work, Badanidiyuru et al.~\cite{BDO11}
gave a $2+\epsilon$ approximation algorithm to the same weighted problem.}.

\begin{center}
\small{}\tt{} \fbox{
\parbox{6.2in}{ \underline{\submax}
\begin{itemize}
\item Let $v^*=\max_{i\in A} v(i)$ and  ${\cal V}=\{v^*,2 v^*,\ldots, n v^*\}$
\item For each $v\in \cal{V}$
    \begin{itemize}
    \item Set $p(i)=\frac{v}{2B}\cdot c(i)$ for each $i\in A$, and find
          $T\in \argmax\limits_{S\subseteq A}\left(v(S)-\sum_{i\in S}p(i)\right)$.
    \item Let $S_v=\emptyset$.
    \item If $v(T)< \frac{v}{2}$, then continue to next $v$.
    \item Else, in decreasing order of $c(i)$ put items from $T$ into $S_v$ while preserving the budget constraint.
    \end{itemize}
\item Output: $S_v$ with the largest value $v(S_v)$ for all $v\in {\cal V}$.
\end{itemize}
}}
\end{center}

\begin{lemma}\label{lemma-SA-max} \submax\ is an $8$-approximation algorithm for subadditive maximization given a demand oracle.
\end{lemma}
\begin{proof}
Let $S^*$ be an optimal solution. Note that $v(S^*)\ge v^*=\max_{i\in A} v(i)$ and $c(S^*)\le B$.
For all $v \le v(S^*)$, we first prove that the algorithm will generate a non-empty set $S_v$ with $v(S_v)\ge \frac{v}{4}$.
Since $T$ is the maximum set returned by the oracle, we have
\[ v(T)-\frac{v}{2B} c(T) \geq  v(S^*)-\frac{v}{2B} c(S^*) \geq  v-\frac{v}{2B} \cdot B \geq  \frac{v}{2}\]
Hence, $v(T)\geq \frac{v}{2}$.
If $c(T)\le B$, then $S_v=T$ and we are done. Otherwise,
by the greedy procedure of picking items from $T$ to $S_v$, we are guaranteed that $c(S_v)\geq \frac{B}{2}$.
Assume for contradiction that $v(S_v)<\frac{v}{4}$.
Then
\begin{eqnarray*}
v(T\setminus S_v)-\frac{v}{2B} c(T\setminus S_v)
&\geq & v(T)-v(S_v)-\frac{v}{2B} \big(c(T) - c(S_v)\big)\\
&> & v(T)-\frac{v}{4}- \frac{v}{2B} c(T) + \frac{v}{2B} \cdot \frac{B}{2} \\
&= & v(T)-\frac{v}{2B} c(T)
\end{eqnarray*}
The later contradicts to the definition of $T$, since $T\setminus S_v$ is then better than $T$.
Thus, we always have $v(S_v)\geq \frac{v}{4}$ for each $v \le v(S^*)$.
Since the algorithm tries all possible $v\in \cal{V}$ (including one with $\frac{v(S^*)}{2}<v\leq v(S^*)$) and outputs
the largest $v(S_v)$, the output is guaranteed to be within a factor of 8 to the optimal value $v(S^*)$.
\end{proof}

Note that we can actually modify the algorithm to get a $(4+\epsilon)$-approximation with running time polynomial in $n$ and $\frac{1}{\epsilon}$. To do so one
may simply replace ${\cal V}$ by a larger set $\big\{\epsilon v^*,2\epsilon v^*,\ldots, \lceil\frac{n}{\epsilon}\rceil
\epsilon v^*\big\}$. Both algorithms suffice for our purpose; for the rest of the paper, for simplicity we will use the 8-approximation
algorithm to avoid the extra parameter $\epsilon$ in the description.

We will use \submax\ as a subroutine to build a mechanism \paymentsharing\ for subadditive functions in the next subsection.
When there are different sets maximizing $v(S)-\sum_{i\in S}p(i)$,
we require that the demand query oracle always returns a fixed set $T$.
This property is important for the truthfulness of our mechanism.
To implement this, we set a fixed order on all the items $i_1\prec i_2\prec \cdots \prec i_n$.
We first compute
\[T_1\in \argmax_{S\subseteq A}\big(v(S)-\sum\limits_{i\in S}p(i)\big) \quad \text{and} \quad T_2\in \argmax_{S\subseteq A\setminus \{i_1\}}\big(v(S)-\sum\limits_{i\in S}p(i)\big).\]
If $v(T_1)-\sum_{i\in T_1}p(i) = v(T_2)-\sum_{i\in T_2}p(i)$, we know that there is a subset without $i_1$ that gives
us the maximum; thus, we ignore $i_1$ for further consideration.
If $v(T_1)-\sum_{i\in T_1}p(i) > v(T_2)-\sum_{i\in T_2}p(i)$, we know that $i_1$ should be in any optimal
solution; hence, we keep $i_1$ and proceed with the process iteratively for $i_{2},i_{3},\ldots,i_n$.
This process clearly gives a fixed outcome that maximizes $v(S)-\sum_{i\in S}p(i)$.

\subsubsection{Mechanism}\label{section-sub-log-mechanism}

Let us first consider the following mechanism based on random sampling and cost sharing.
\begin{center}
\small{}\tt{} \fbox{
\parbox{6.4in}{\hspace{0.05in} \\
\underline{\paymentsharing}
\begin{enumerate}
\item Pick each item independently at random with a probability of $\frac{1}{2}$ into group $T$.
\item Run \submax\ for items in group $T$, and let $v$ be the value of the returned subset.
\item For $k=1$ to $|A\setminus T|$
\begin{itemize}
  \item Run \submax\ on the set $\left\{i\in A\setminus T~|~c(i)\le \frac{B}{k}\right\}$
  where each item has cost $\frac{B}{k}$, denote the output by $X$.
  \item If $v(X)\ge\frac{\log\log n}{80\log n} \cdot v$
  \begin{itemize}
  \item Output $X$ as the winning set and pay $\frac{B}{k}$ to each item in $X$.
  \item Halt.
  \end{itemize}
\end{itemize}
\item Output $\emptyset$.
\end{enumerate}
}}
\end{center}

In the above mechanism, we again first sample in expectation half of the items to form a testing group $T$,
and then use \submax\ to compute an approximate solution for the items in $T$ given the budget constraint $B$.
As can be seen in the analysis of the mechanism, the computed value $v$ is in expectation within a constant factor
of the optimal value of the whole set $A$.
That is, we are able to learn the rough value of the optimal solution by random sampling.
Next we consider the remaining items $A\setminus T$ and try to find a subset $X$ with a relatively big value in which
every item is willing to ``share'' the budget $B$ at a fixed share $\frac{B}{k}$.
(This part of our mechanism can be viewed as a reversion of the classic cost sharing mechanism.)
Finally, we use the information $v$ from random sampling as a benchmark to determine
whether $X$ should be a winning set or not.

The final mechanism for subadditive functions is described as follows.

\begin{center}
\small{}\tt{} \fbox{
\parbox{6.4in}{\hspace{0.05in} \\
\underline{\submain}
\begin{itemize}
\item With half probability, run \paymentsharing.
\item With half probability, pick a most-valuable item as the only winner and pay him $B$.
\end{itemize}
} }
\end{center}

\begin{theorem}\label{theorem-SA-loglog}
\submain\ runs in polynomial time given a demand oracle and is a truthful budget feasible mechanism for subadditive functions with an
approximation ratio of $O(\frac{\log n}{\log \log n})$.
\end{theorem}
The proof of the Theorem~\ref{theorem-SA-loglog} is given in the appendix.

\newcommand{\SubaddGAP}{{\sc SA-mechanism-Integrality-Gap}}
\newcommand{\SubaddBayesian}{{\sc SA-Bayesian-mechanism}}
\newcommand{\nf}{{\widetilde{f}}}
\newcommand{\vecd}{{\mathbf d}}

\section{Bayesian Mechanism Design}\label{section-bayesian}

In this section, we study budget feasible mechanisms for subadditive functions from a standard economics viewpoint,
where the costs of all agents $(c(i))_{i\in A}$ are drawn from a prior known distribution $\calD$.
More specifically, the mechanism designer and all participants know $\calD$ in advance from which the real cost vector $(c(i))_{i\in A}$ is drawn.
However, each $c(i)$ is the private information of agent $i$. Distribution $\calD$ is given on the probability space $\Omega$ with the corresponding density function $\rho(\cdot)$ on $\R^{|A|}$.
We allow dependencies on the agents' costs in $\calD$ and consider the distributions that have integrable marginal
densities for any subset of variables\footnote{\small We need some mild technical restriction on $\calD$
in order to sample a conditioned random variable. We assume that the density function $\rho(\cdot)$ of $\calD$
is integrable over each subset $S\subseteq A$ of its variables for any choice of the rest parameters, i.e., $\rho(c_{A\setminus S})=\int_{\Omega}\ \rho(c)\, dx_{S}$ is bounded. This condition is
reminiscent of integrability of marginal density functions (see, e.g., page 331 of~\cite{probability-book}), though in our case we require a slightly stronger condition.};
this includes, e.g., independent distributions as special cases.

Every agent submits a bid $b(i)$ as before and seeks to maximize his own utility.
We again consider universally truthful mechanisms, i.e., for every coin flips of the mechanism and each cost vector, truth-telling is
a dominant strategy for every agent.
The performance of a mechanism $\calM$ is measured by $\ex[\calM]=\ex_{c\sim \calD}[\mathcal{M}(c)]$.
We compare a mechanism with the optimal expected value $\ex[\opt]=\ex_{c\sim \calD}\big[v(\opt(c))\big]$;
we say mechanism $\calM$ is a (Bayesian) $\alpha$-approximation if $\frac{\ex[\opt]}{\ex[\calM]}\le \alpha$.

In this section: Let $\opt_{v}(c,S)$ denote the winning set in an optimal solution when
the valuation function is $v(\cdot)$, the cost vector is $c$, and the agent set is $S$ (the parameters are omitted if they
are clear from the context);
let $v(\opt_{v}(c,S))$ denote the value of $\opt_{v}(c,S)$.

Our mechanism is as follows.

\begin{center}
\small{}\tt{} \fbox{
\parbox{6.5in}{\hspace{0.05in}
\underline{\SubaddBayesian}
\begin{itemize}
\item With a probability of $\frac{1}{2}$, let a most-valuable item be the only winner and pay him $B$.
\item With a probability of $\frac{1}{2}$, run the following:
\begin{enumerate}
\item Pick each item independently at random with a probability of $\frac{1}{2}$ into group $T$.
\item Compute an optimal solution $\opt(c,T)$ for items in $T$ given budget $B$.
\item Set a threshold $t=\frac{v(\opt(c,T))}{8B}.$
\item For items in $A\setminus T$ find a set $S^*\in \argmax\limits_{S\subseteq A\setminus T}\big\{v(S)-t\cdot c(S)\big\}.$
\item \label{cond} Sample a cost vector $d\sim\calD$ conditioned on
      \begin{enumerate}
      \item\label{cond_one} $d(i)=c(i)$ for each $i\in T$, and
      \item\label{cond_two} $S^*\in\argmax\limits_{S\subseteq A\setminus T}\big\{v(S)-t\cdot d(S)\big\}.$
      \end{enumerate}
\item \label{less_case} If $d(S^*)< B$, let all $i\in S^*$ with $c(i)\le d(i)$ be the winners.
\item \label{more_case} If $d(S^*)\ge B$,
\begin{itemize}
\item run \XOSmain\ w.r.t.~valuation $\nv(\cdot)$, set $S^*$, cost $c(\cdot)$, budget $B$.
\item Output the result of \XOSmain.
\end{itemize}
\end{enumerate}
\end{itemize}
} }
\end{center}

In the mechanism, Steps~(1-3) are the same as \XOSsample\ where we randomly sample a test group $T$ and generate a threshold value $t$.
In Steps~(4-7), we consider a specific subset $S^*\subseteq A\setminus T$ and select winners only inside of it.
Step~(5) helps to give us a guidance on the threshold payments of the winners (see more discussions below).
Step~(7) runs \XOSmain\ on the function $\nv(\cdot)$ (defined as the optimal value of the LP ($\lozenge$)), which is XOS according to Lemma~\ref{lemma-nv-XOS}.

A few remarks about the mechanism are in order.
\begin{itemize}
\item It is tempting to remove the random sampling part, as given $\calD$ one may consider a `prior sampling'
approach: Generate some virtual instances according to $\calD$ and compute a
threshold $t$ based on them; then apply this threshold to all agents in $A$. Interestingly, the
prior sampling approach works well in our mechanism when, e.g., all $c(i)$'s are independent,
but it does not work for the case when variables are dependent.


For instance, consider an additive valuation $v(\cdot)$ with $v(S)=|S|$, budget $B=2^k$ for a large $k$,
and a set of $N=2^k$ agents with the following discrete distribution over costs ($c=\ell$ means
that every $c(i)=\ell$):
\[
\pr[c=1] = \frac{1}{2^{k+1}}, \pr[c=2]=\frac{1}{2^{k}}, \ldots\ldots,
\pr\big[c=2^k\big]=\frac{1}{2}, \pr\big[c=2^{k+1}\big]=\frac{1}{2^{k+1}}.
\]
Note that
\[
v(\opt(c=1)) = 2^k, v(\opt(c=2)) = 2^{k-1}, \ldots\ldots,
v\big(\opt\big(c=2^k\big)\big) = 1, v\big(\opt\big(c=2^{k+1}\big)\big)=0 .
\]
%
%
    Then the expected optimal value is $\ex[\opt]=\frac{k+1}{2}$ and it is equally spread over all possible costs except the last one $c=2^{k+1}$. Roughly speaking, on a given instance $c$, any prior estimate on $v(\opt(c))$ that gives a constant approximation only applies to a constant number of distinct costs (the contribution of these cases to $\ex[\opt]$ is negligible). Hence for almost all other possible costs, we get a meaningless estimate for $\opt(c)$.
    Therefore, the prior sampling will lead to a bad approximation ratio.

\item Why do we generate another cost vector $d$ in Step~(5)? Recall that our target winner set is $S^*$, whose value $v(S^*)$ in expectation
gives a constant approximation of $\ex[\opt]$. However, we are faced with the problems of selecting a winning set in $S^*$
with a sufficiently large value and distributing the budget among the winners. These two problems together are closely related to
cooperative game theory and the notion of approximate core. For subadditive functions, a constant approximate core may not exist~\cite{agt-book} (e.g., set cover gives a logarithmic lower bound~\cite{BR11}). Thus we might not be able to pick a winning set with a constant approximation and set threshold payments in accordance with the valuation function.
The question then is: Is there any other guidance we can take to bound budget feasible threshold payments and give a constant approximation?


Our solution is to use another random vector $d$ to serve as such a guidance. (Conditions in Steps~(\ref{cond_one}) and (\ref{cond_two}), from a high level point of view, guarantee that the vector $d$ is not too `far' from $c$ for the agents in $S^*$, in the sense that both vectors are derived from the same distribution.
Thus, cost vectors $c$ and $d$ are distributed symmetrically and
can be switched while preserving some important parameters such as $t$ and $S^*$ in expectation.)
If $d(S^*)\le B$ (Step~(6)), then we set $d(i)$ as an upper bound on the payment of each agent $i\in S^*$;
this guarantees that we are always within the budget constraint.
If $d(S^*)> B$, setting $d(i)$ as an upper bound is not sufficient to ensure budget feasibility;
then we adopt our approach for XOS functions with inputs subset $S^*$ and XOS valuation $\nv(\cdot)$
defined by $(\lozenge)$.



\end{itemize}

\begin{theorem}\label{theorem-bayesian}
\SubaddBayesian\ is a universally truthful budget feasible mechanism for subadditive functions and gives in expectation a constant approximation.
\end{theorem}

Budget feasibility follows simply from the description of the mechanism and the fact that \XOSmain\ is budget feasible.

For universal truthfulness, we note that in the mechanism, the sampled vector $d$ comes from a
distribution that depends on actual bid vector $c$. To see why our mechanism takes a distribution over deterministic truthful mechanisms,
we can describe all possible samples $d$ for (i) all possible cost vectors on $T$ and (ii) all possible choices
$S\subseteq A\setminus T$ of $S^*$; then we tell all flipped $d$'s to the agents
before looking at the costs of $T$. (Practically, we can provide all our randomness as a black box accessible by all agents.)
Note that the selection rule of $S^*$ is monotone, and, similarly to Claim~\ref{lem:SameS}, each agent in $S^*$
cannot manipulate (i) the composition of $S^*$ given $c$ and $T$, and (ii) the choice of $d$, as long as he stays in $S^*$.
Therefore, composing the first part choosing $S^*$ (Step~(4)) with the next monotone rule picking winners in $S^*$ (Steps~(6-7)),
we again get a monotone winner selection rule. Hence, the mechanism is universally truthful.

Next we give a sketch of the idea of proving the constant approximation. Details of the proof are deferred to the Appendix~\ref{section-bayesian-proof}. 

\medskip
\noindent {\em Approximation analysis (sketch).}
We sketch the proof idea of the approximation ratio of the mechanism.
First, similar to our analysis in Section~\ref{sec:xos},
the optimal solution $v(\opt(c,T))$ obtained from random sampling in expectation gives a constant approximation to the optimal solution $\ex[\opt]$.
Further, we observe the following facts (which are reminiscent of Claim~\ref{lem:subset}):
$$\nv(S) - t\cdot c(S) \geq 0 \ \ \mbox{and} \ \ \nv(S) - t\cdot d(S) \geq 0, \ \ \forall S \subseteq S^*$$
where the second inequality is based on the conditional distribution we choose for $d$.

If $c(S^*)\ge B$ and $d(S^*)\ge B$ (i.e., the mechanism runs Step~(\ref{more_case})),
we can pick a subset $S^0\subseteq S^*$ with $B \ge c(S^0)\ge \frac{B}{2}$.
By Theorem~\ref{theorem-XOS-main}, \XOSmain\ gives a constant approximation to the optimum of $\nv(\cdot)$ on $S^*$.
(This is the reason why in Step~(7) of the mechanism, we run the whole \XOSmain\ on the input instance $\nv(\cdot)$ and $S^*$.)
Hence,
$$\nv(\opt_{\nv}(c, S^*))\ge \nv(S^0)\ge t\cdot c(S^0)\ge t\cdot \frac{B}{2}\ge \frac{v(\opt(c,T))}{16},$$
where the first inequality follows from the fact that $S^0\subseteq S^*$ is a budget feasible set.
Thus, the optimum of $\nv(\cdot)$ on $S^*$ is within a constant factor of $v(\opt(c,T))$, as well as the benchmark $\ex[\opt]$.

If $c(S^*)<B$ and $d(S^*)\ge B$, we have $\nv(S^*)\ge t\cdot d(S^*)\ge\frac{v(\opt(c,T))}{8}.$
Further, we notice that $S^*$ is budget feasible with respect to cost vector $c$; thus, \XOSmain\ gives
a constant approximation to $\nv(S^*)$, which in turn is within a constant factor of $v(\opt(c,T))$ and $\ex[\opt]$.

We observe that the vectors $d$ and $c$ are restricted to the agents in $S^*$ and conditioned on
\[
S^*\in \argmax\limits_{S\subseteq A\setminus T}\big\{v(S)-t\cdot d(S)\big\} \quad \text{and}\quad
S^*\in \argmax\limits_{S\subseteq A\setminus T}\big\{v(S)-t\cdot c(S)\big\}
\]
have exactly the same distributions.
Therefore, due to such a symmetry between $d$ and $c$, in a run of
our mechanism in expectation we will have the outcome $T, t, S^*$ and a pair of vectors $(c,d)$ as often as the outcome
$T, t, S^*$ and the pair $(d,c)$. This implies that in the case when $d(S^*)<B$ and $c(S^*)<B$, we
get on average a value of at least $\frac{1}{2}v(S^*)$, since (i) the winning sets on the two instances where
$c$ (resp., $d$) is the private cost and $d$ (resp., $c$) is the sampled cost altogether cover $S^*$, and (ii) $v$ is a subadditive function.
By the choice of threshold $t$, we also know that
\begin{eqnarray*}
v(S^*) &\ge& v(S^*)- t\cdot c(S^*) \\
&\ge& v(\opt(c,A\setminus T))- t\cdot c(\opt(c,A\setminus T)) \\
&\ge& v(\opt(c,A\setminus T))- t\cdot B.
\end{eqnarray*}
Thus, our mechanism gives a constant approximation to $v(\opt(c))$ with some constant probability.

The last case is when $c(S^*)\ge B$ and $d(S^*)<B$. Again due to the symmetry between $c$ and $d$, intuitively, we can
treat this case as the above one when $c$ and $d$ are switched; thus we also get a constant approximation of $\ex[\opt]$.
(The formal argument, however, due to multiple randomness used in the mechanism, is much more complicated.)

Therefore, the mechanism \SubaddBayesian\ on average has a constant approximation of the expected socially optimal value $\ex[\opt]$. \hfill $\square$

\section{Conclusions}

Our work considers budget feasible mechanism design in two analysis frameworks: prior-free and Bayesian.
For XOS functions, we give a prior-free constant approximation mechanism.
For subadditive functions, we present two prior-free mechanisms with integrality-gap and sub-logarithmic
approximations, respectively, as well as a Bayesian constant approximation mechanism.
All our mechanisms are universally truthful.

Our mechanisms continue to work for the extension when the valuation functions are non-monotone,
i.e., $v(S)$ is not necessarily less than $v(T)$ for any $S\subset T\subseteq A$.
For instance, the cut function studied in~\cite{DPS11} is non-monotone.
For such functions, we can define $\hat{v}(S)=\max_{_{T\subseteq S}}v(T)$ for any $S\subseteq A$.
It is easy to see that $\hat{v}(\cdot)$ can be computed easily given a demand oracle,
is monotone, and inherits the classification of $v(\cdot)$.
Further, any solution maximizing $v(\cdot)$ is also an optimal solution of $\hat{v}(\cdot)$.
Hence, we can apply our mechanisms to $\hat{v}(\cdot)$ directly and obtain the same approximations.


We give a constant approximation mechanism for subadditive functions in the Bayesian framework where the costs are drawn from a given {\em known} distribution.
Considering the gap between Bayesian and prior-free, a natural question is under which frameworks a constant approximation mechanism still exists.
An interesting step along this direction is the prior-independent setting, where the costs are still drawn from an underlying distribution, but the mechanism designer does not have the prior knowledge of it.
Our mechanism \SubaddBayesian\ can be adopted to the framework where all the costs are identically and independently distributed.
(Specifically, by random sampling we are able to learn the underlying unknown distribution with sufficient precision.)
However, we do not know how to handle independent but not necessarily identical distributions in the
prior-independent framework, as well as the most general prior-free setting in the worst-case analysis.

Indeed, whether subadditive functions admit a prior-free constant approximation mechanism still remains an open problem.
Our results show a separation between XOS and subadditive functions. Another angle to have such a distinction between the two classes is from exponential concentration:
In the case of XOS the valuation of a randomly selected subset obeys an exponential concentration around its expected value\footnote{\small Note that this fact can be used to improve the approximation ratios of the mechanisms of XOS functions substantially (but still up to a constant factor).}, whereas in the case of general subadditive
valuations it does not (see~\cite{vondrak} for a counterexample).
Such a difference on exponential concentration may suggest a possible distinction between XOS and subadditive functions in terms of their approximability in (budget feasible) mechanism design (see appendix~\ref{ap:concentration} for more details).

For those mechanisms with exponential runtime, it is natural to ask if there are truthful designs with the same approximations that can be implemented in polynomial time. Further, all of our mechanisms are randomized; it is intriguing to consider the approximability of deterministic mechanisms. We leave these questions as future work.

%
%

\section*{Acknowledgements}

We thank the anonymous reviewers for their helpful comments.
We are grateful to Jason Hartline for many valuable suggestions on exploiting intuitions of the mechanisms (in particular, the connection to Lagrangian).

\input{appendix}

\end{document}

%% file: appendix.tex
\appendix

\section{Proof of Theorem~\ref{theorem-SA-loglog}}

\begin{proof}
Let $S=A\setminus T$.
It is obvious that the mechanism runs in polynomial time since \submax\ is in polynomial time.
If the mechanism picks the largest item, certainly it is budget feasible as the total payment is precisely $B$.
If it chooses \paymentsharing, either no item is a winner or $X$ is selected as the winning set. Note that $|X|\le k$ and each item in $X$ gets a payment of $\frac{B}{k}$. It is therefore budget feasible as well.

\medskip \noindent
{\em (Truthfulness.)} Truthfulness for picking the largest item is obvious (as the outcome is irrelevant to the submitted bids). Next we will prove that \paymentsharing\ is truthful as well.
The random sampling step does not depend on the bids of the items, and items in $T$ have no incentive to lie as they cannot win anyway. Hence, it suffices to only consider items in $S$.
Observe that every agent will be a candidate to the winning set only if $c(i)\le \frac{B}{k}$.
Consider any item $i\in S$ and fixed bids of other items. There are the following three possibilities if $i$ reports his true cost $c(i)$.
\begin{itemize}
  \item Item $i$ wins with a payment $\frac{B}{k}$. Then we have $c(i)\le \frac{B}{k}$ and his utility is $\frac{B}{k}-
  c(i)\geq 0$. If $i$ reports a bid which is still less than or equal to $\frac{B}{k}$, the
  output and all the payments do not change. If $i$ reports a bid which is larger than $\frac{B}{k}$, he still could not
  win for a share larger than $\frac{B}{k}$ and will not be considered for all smaller shares. Therefore, he derives $0$
  utility. Thus for either case, $i$ does not have incentive to lie.
  \item Item $i$ loses and payment to each winner is $\frac{B}{k}\geq c(i)$. In this case, if $i$ reduces or increases
  his bid, he cannot change the output of the mechanism. Thus $i$ always has zero utility.
  \item Item $i$ loses and payment to each winner is $\frac{B}{k} < c(i)$ or the winning set is empty.
   In this case, if $i$ reduces his bid, he will not change the process of the mechanism until the payment offered by the
   mechanism is less than $c(i)$. Thus, even if $i$ could win for some value $k$, the payment he gets would be less than
   $c(i)$, in which case his utility is negative. If $i$ increases his bid, he lose and thus derives
   zero utility.
\end{itemize}
Therefore, \paymentsharing\ is a universally truthful mechanism.

\medskip \noindent
{\em (Approximation Ratio.)}
It remains to estimate the approximation ratio. Let $\opt=\opt(A)$ denote the optimal solution for the whole set.

If there exists an item $i\in A$ such that $v(i)\geq \frac{1}{2} v(\opt)$, then picking the largest item will generate a value at least $\frac{1}{2} v(\opt)$ and we are done. In the following, we assume that $v(i)<\frac{1}{2} v(\opt)$ for all $i\in A$. Then, by Lemma \ref{lem:probability}, with probability at least $\frac{1}{2}$ we have $v(\opt(T)) \geq \frac{1}{8} v(\opt)$ and $v(\opt(S)) \geq \frac{1}{8} v(\opt)$. Hence, with probability at least $\frac{1}{4}$ we have
\begin{equation}\label{eq:separation}
v(\opt(S)) \geq  v(\opt(T)) \geq\frac{1}{8} v(\opt).
\end{equation}
Therefore, it suffices to prove that the main mechanism has an
approximation ratio of $O(\frac{\log n}{\log \log n})$ given the
inequalities \eqref{eq:separation}.

Since \submax\ is an $8$ approximation of $v(\opt(T))$, we have
$v\geq \frac{1}{8}v(\opt(T)) \geq \frac{1}{64}v(\opt)$. Clearly, if
\paymentsharing\ outputs a non-empty set, then its value is at least
$\frac{\log\log n}{80\log n} \cdot v\geq \frac{\log\log n}{5120 \log
n} \cdot v(\opt)$. Hence, it remains to prove that the mechanism will
always output a non-empty set given formula~(\ref{eq:separation}).

Let $S^*=\{1,2,3,\ldots, m\}\subseteq S$ be an optimal solution of $S$ given the budget constraint $B$ and $c_1\geq c_2 \geq \cdots \geq c_m$. We recursively divide the agents in $S^*$ into different groups as follows:
\begin{itemize}
  \item Let $\alpha_1$ be the largest integer such that $c_1 \leq \frac{B}{\alpha_1}$.
        Put the first $\min\{\alpha_1,m\}$ agents into group $Z_1$.
  \item Let $\beta_r=\alpha_1+\dots+\alpha_r$. If $\beta_r < m$ let $\alpha_{r+1}$ be the largest integer such that $c_{_{\beta_r+1}} \leq \frac{B}{\alpha_{r+1}}$;
        put the next $\min\{\alpha_{r+1},m-\beta_r\}$ agents into group $Z_{r+1}$.
\end{itemize}

Let us denote by $x+1$ the number of groups.
Since items in $S^*$ are ordered by $c_1\geq c_2 \geq \cdots \geq c_m$, we have $\alpha_{r+1}\ge \alpha_r$ for any $r$.
If there exists a set $Z_j$ such that $v(Z_j)\geq\frac{\log\log n}{10\log n}\cdot v$,
then the mechanism does not output an empty set, as it could buy $\alpha_j$ items at price $\frac{B}{\alpha_j}$ given that \submax\ is an
$8$-approximation and the threshold we set is $v(Z_j)\geq\frac{\log\log n}{80\log n}\cdot v$.
Thus, we may assume that $v(Z_j) < \frac{ \log \log n}{10 \log n} \cdot v$ for each $j=1,2,\ldots, x+1$.
On the other hand, by subadditivity, we have
\[\sum_{j=1}^{x+1} v(Z_j) \ge v(S^*) = v(\opt(S)) \geq v(\opt(T))\ge v.\]
Putting the two inequalities together, we can conclude that $(x+1)\cdot \frac{\log\log n}{10\log n}\cdot v > v$, which implies that
\[x > \frac{5\log n}{\log \log n} \geq \frac{5\log m}{\log \log m}. \]
%

On the other hand, since $S^*=\{1,2,3,\ldots, m\}$ is a solution for $S$ within the budget constraint, we have that $\sum_{j=1}^{m} c_j \leq B$. Further, since
$c_1 > \frac{B}{\alpha_1+1}, c_{_{\beta_1+1}} > \frac{B}{\alpha_2+1}, \ldots, c_{_{\beta_x+1}} > \frac{B}{\alpha_{x+1}+1}$, we have
\begin{eqnarray*}
B &\geq& \sum_{j=1}^{m} c_j \geq c_1 + \alpha_1 c_{_{\beta_1+1}} +\cdots +\alpha_x c_{_{\beta_x+1}} \\
&>&
\frac{B}{\alpha_1+1} + \frac{\alpha_1 B}{\alpha_2+1} +\cdots +\frac{\alpha_x B}{\alpha_{x+1}+1}.
\end{eqnarray*}
Hence,
\[
1 \geq \frac{1}{\alpha_1+1} + \frac{\alpha_1 }{\alpha_2+1} +\cdots +\frac{\alpha_x }{\alpha_{x+1}+1} \geq
\frac{1}{2\alpha_1} + \frac{\alpha_1 }{2\alpha_2} +\cdots +\frac{\alpha_x }{2\alpha_{x+1}}.
\]
In particular, we get
\[
2 \geq \frac{1}{\alpha_1} + \frac{\alpha_1 }{\alpha_2} +\cdots +\frac{\alpha_{x-1} }{\alpha_x}\geq x\sqrt[^x]{\frac{1}{\alpha_1}\frac{\alpha_1 }{\alpha_2}\cdots\frac{\alpha_{x-1} }{\alpha_x}},
\]
where the last inequality is simply the inequality of arithmetic and geometric means.
Hence, we get $2\ge x\sqrt[^x]{\frac{1}{\alpha_x}}$, which
is equivalent to $\alpha_x\ge (\frac{x}{2})^{x}$. Now plugging in the fact that $m\ge \alpha_x$ and $x\ge\frac{5\log m}{\log \log m}$, we come to a contradiction.
This concludes the proof.
\end{proof}

\section{Analysis of \SubaddBayesian}\label{section-bayesian-proof}

In this section, we will give a formal proof for the constant approximation of \SubaddBayesian.
We first have the following observation.

\begin{claim}\label{cl:subset}
For any $S \subseteq S^*$, $\nv(S) - t\cdot c(S) \geq 0$.
\end{claim}
\begin{proof}
Indeed, recall that $S^*=\argmax\limits_{S\subseteq A\setminus T}\big\{v(S)-t\cdot c(S)\big\}$.
Then $v(S)-t\cdot c(S)\ge 0$ for any $S\subseteq S^*$, since otherwise, we have
$v(S)-t\cdot c(S)<0$ for some $S\subseteq S^*$ and we get
\small
\begin{eqnarray*}
  v(S^*\setminus S) - t\cdot c(S^*\setminus S) & \ge & v(S^*) - v(S) - t\cdot c(S^*\setminus S) \\
  & = & v(S^*) - t\cdot c(S^*) - \big(v(S) - t\cdot c(S)\big) \\
  & > & v(S^*) - t\cdot c(S^*),
\end{eqnarray*}
\normalsize
a contradiction.

Thus, for any $S\subseteq S^*$, $v(S)\ge t\cdot c(S)$. Therefore, in the description of
$LP(S)$, we have $v(S_j)\ge t\cdot c(S_j)$ for each $j\in[1,N]$. Now substituting $v(S_j)$ for $c(S_j)$ in
$LP(S)$, we get the desired inequality $\nv(S)=LP(S)\ge t\cdot c(S).$
\end{proof}

We assume that no single item can have cost more than $B$ in the cost vector $c$.
We need some extra notations. We assume that the distribution $\calD$ is given on the probability
space $\Omega$ with corresponding density function $\rho(x)$ on $\R^{|A|}$. For a set $T\subseteq A$, let $x_{_T}$ be a
point in $\R^{|T|}$. We will denote by $\rho(x_{_T})$ a distribution's density we get on the corresponding space $\R^{|T|}$ by sampling
$x\in \R^{|A|}$ from $\calD$ and restricting it to $T$-coordinates of $x$. By $\rho(x|x_{_T})$ we denote the conditional
density, if we fix $T$-coordinates of $x$ to be the same as in $x_{_T}$. By $\opt(x,S)$ we denote the optimal
value we can get from the set $S\subseteq A$ on cost vector $x$ with budget $B$.
In order to shorten the notations sometimes we will omit $x$ or $S$, in the latter case $S=A$;
sometimes we take the optimum over valuation $\nv$ instead of $v$, so to emphasis this we
employ notation $\opt_{\nv}(x,S).$ In the mechanism we compute set $S^*(c,T)$,
which depends on the sampled set $T$ and cost vector $c$. By $XOS(c,S^*)$
we denote the value we get from \XOSmain$(c,S^*)$ run on the set $S^*$, cost vector $c$ and
XOS valuation $\nv(\cdot)$ that solely depends on $v(\cdot)$.

In the following we write explicitly the expected value of the
second part of our mechanism.

\begin{equation}
\int\limits_{\Omega}\frac{1}{2^{|A|}}\sum_{T\subset A}~~\int\limits_{\Omega}
f(x,y,T)\rho\Big(y \left| y=x\left|_{_{T}}; S^*(x,T)=S^*(y,T)\right.\Big)\right.\mathrm{d}y\rho(x)\mathrm{d}x,
\label{eq:exp_value}
\end{equation}
where
\[
f(c,d,T)=
  \begin{cases}
   XOS(c,S^*) & \text{if } d(S^*)\ge B \\
   v(S^*\cap \{i: c(i)\le d(i)\}) & \text{if } d(S^*)< B,
  \end{cases}
\]

Swapping in \eqref{eq:exp_value} the sum and integral we get

\begin{align*}
\frac{1}{2^{|A|}}\sum_{T\subset A}~~\int\limits_{\Omega}\int\limits_{\Omega}
f(x,y,T)\rho\Big(y \left| y=x\left|_{_{T}}; S^*(y,T)=S^*(x,T)\right.\Big)\right.\mathrm{d}y\rho\Big(x\Big)\mathrm{d}x=\\
\frac{1}{2^{|A|}}\sum_{T\subset A}~~\int\limits_{\Omega(T)}\int\limits_{\Omega}\int\limits_{\Omega}
f(x,y,T)\rho\Big(y \left| x_{_{T}}; S^*(y,T)=S^*(x,T)\Big)\right.\mathrm{d}y~ \rho\Big(x\left|x_{_T}\Big)\right.\mathrm{d}x~\rho\Big(x_{_T}\Big)\mathrm{d}x_{_T}=\\
\frac{1}{2^{|A|}}\sum_{T\subset A}\int\limits_{\Omega(T)}\int\limits_{\Omega}\sum_{S\subset A\setminus T}
\left[\int\limits_{\Omega}f(x,y,T)\rho\Big(y \left| x_{_{T}}; S^*(y,T)=S^*(x,T)\Big)\right.\mathrm{d}y\right]\\
\rho\Big(x\left|x_{_T}; S^*(x,T)=S\Big)\right.\cdot Pr\Big(S^*(x,T)=S\left|x_{_T}\Big)\right.\mathrm{d}x~\rho\Big(x_{_T}\Big)\mathrm{d}x_{_T}=\\
\frac{1}{2^{|A|}}\sum_{T\subset A}\int\limits_{\Omega(T)}\sum_{S\subset A\setminus T}
\left[\int\limits_{\Omega}\int\limits_{\Omega}f(x,y,T)\rho\Big(y \left| x_{_{T}}; S^*(y,T)=S\Big)\right.\mathrm{d}y\cdot
\rho\Big(x\left|x_{_T}; S^*(x,T)=S\Big)\right.\mathrm{d}x\right]\\
Pr\Big(S^*(\cdot,T)=S\left|x_{_T}\Big)\right.~\rho\Big(x_{_T}\Big)\cdot\mathrm{d}x_{_T}=\\
\frac{1}{2^{|A|}}\sum_{T\subset A}\int\limits_{\Omega(T)}\sum_{S\subset A\setminus T}
\int\limits_{\Omega}\int\limits_{\Omega}\frac{f(x,y,T)+f(y,x,T)}{2}\rho\Big(y \left| x_{_{T}}; S^*(y,T)=S\Big)\right.\mathrm{d}y\\
\rho\Big(x\left|x_{_T}; S^*(x,T)=S\Big)\right.\mathrm{d}x~
Pr\Big(S^*(\cdot,T)=S\left|x_{_T}\Big)\right.~\rho\Big(x_{_T}\Big)\cdot\mathrm{d}x_{_T}.
\end{align*}

To get the first equality we split the integral w.r.t. variable $x$ into two integrals w.r.t. variables $x_{_T}\in
\mathbb{R}^{|T|}$ and $x$ conditioned on $x=x_{_T}$ on $T$; the second equality follows from the low of total
expectation applied to the variable in the square brackets with conditioning on all possible values of random
variable $\{S^*(\cdot,T)| x_{_T}\}$; to get the third equality we swap integral with the sum; to get the
last equality we have used the symmetry between $x$ and $y$ for the expression in square brackets.
Next, the formula \eqref{eq:exp_value} can be written as

\begin{equation}
\int\limits_{\Omega}\frac{1}{2^{|A|}}\sum_{T\subset A}~~\int\limits_{\Omega(x,T)}
\frac{f(x,y,T)+f(y,x,T)}{2}\rho\Big(y \left| y=x\left|_{_{T}}; S^*(x,T)=S^*(y,T)\right.\Big)\right.\mathrm{d}y\rho(x)\mathrm{d}x,
\label{eq:sym_expect}
\end{equation}

Now we estimate the value of $XOS(c,S^*)$ in the case when $d(S^*)\ge B$.
Recall that $\nv(S)\le v(S)$ for any set $S\subseteq A$. According to the results of previous
section for the XOS functions we have a constant approximation to the optimum:
$XOS(c,S^*)\ge\alpha\cdot\nv(\opt_{\nv}(c,S^*)),$
for a positive constant $\alpha.$ In the next we consider two cases for $c(S^*).$

\begin{enumerate}
\item $c(S^*)<B.$ Since one can buy the whole set $S^*$, we have
$\opt_{\nv}(c, S^*)=S^*$. We have $\nv(S^*)\ge t\cdot d(S^*)$ by Claim \ref{cl:subset}
applied to the cost vector $d$. Hence, recalling the definition of $t$ we get $\nv(\opt_{\nv}(c,S^*))\ge \frac{v(\opt(c,T))}{8B}B=\frac{v(\opt(c,T))}{8}.$

\item $c(S^*)\ge B.$ Since any single item in $S^*$ has cost less than $B$,
we can find a set $S_0\subset S^*$, such that $B\ge c(S_0)\ge \frac{B}{2}.$ We observe that $\nv(\opt_{\nv}(c,S^*))\ge\nv(S^0)$,
as $S_0$ is a budget feasible set for the cost vector $c$. Then we have $\nv(S^0)\ge t\cdot c(S^0)$ by Claim \ref{cl:subset} and  $\nv(\opt_{\nv}(c, S^*))\ge\nv(S^0)\ge t c(S^0)\ge\frac{v(\opt(c,T))}{16}$.
\end{enumerate}

Thus
\[
f(c,d,T)=XOS(c,S^*)\ge\frac{\alpha}{16}v(\opt(c,T)),
\]
if $d(S^*)\ge B.$ By symmetry between $c$ and $d$ we have
$XOS(d,S^*)\ge\frac{\alpha}{16}v(\opt(d,T)),$ if $c(S^*)\ge B.$
We note that $c$ coincides with $d$ on $T$, and hence $\opt(c,T)=\opt(d,T).$
Therefore,
\[
f(d,c,T)=XOS(d,S^*)\ge\frac{\alpha}{16}v(\opt(c,T)),
\]
if $c(S^*)\ge B.$

In the following we estimate $f(c,d,T)+f(d,c,T).$
\begin{enumerate}
\item If either $c(S^*)\ge B$ or $d(S^*)\ge B$. We observe that due to the last inequalities for $f(c,d,T)$ and $f(d,c,T)$
 \[f(c,d,T)+f(d,c,T)\ge\frac{\alpha}{16}v(\opt(c, T)).
 \]
\item If $c(S^*)< B$ and $d(S^*)< B$, then by subadditivity of $v$ we get
\[
f(c,d,T)+f(d,c,T)= v(S^*\cap \{i: c(i)\le d(i)\})+v(S^*\cap \{i: c(i)\ge d(i)\})\ge v(S^*).
\]
\end{enumerate}
Hence we can write a lower bound on $f(c,d,T)+f(d,c,T)$

\[
f(c,d,T)+f(d,c,T)\ge\min\left(v(S^*),\frac{\alpha}{16}v(\opt(c,T))\right),
\]
which does not depend on $d$. We plug in this lower bound instead of $f(c,d,T)+f(d,c,T)$
into the formula \eqref{eq:sym_expect}.

\begin{align*}
\int\limits_{\Omega}\frac{1}{2^{|A|}}\sum_{T\subset A}~~\int\limits_{\Omega(x,T)}
\frac{\min\left(v(S^*),\frac{\alpha}{16}v(\opt(x,T))\right)}{2}\rho\Big(y \left| y=x\left|_{_{T}}; S^*(x,T)=S^*(y,T)\right.\Big)\right.\mathrm{d}y\rho(x)\mathrm{d}x=\\
\int\limits_{\Omega}\frac{1}{2^{|A|}}\sum_{T\subset A}
\min\left(\frac{v(S^*(x,T))}{2},\frac{\alpha}{32}v(\opt(x,T))\right)\rho(x)\mathrm{d}x.
\end{align*}
The equality holds, since $y$ comes from probability distribution and the function under
integral does not depend on $y$.

Let $i^*(x)$ be the most valuable item in $A$ with $c(i^*)\le B$.
We can write the following lower bound on the total valuation of \SubaddBayesian.

\[
\int\limits_{\Omega}\left(0.5\times v(i^*(x))+0.5\times\frac{1}{2^{|A|}}\sum_{T\subset A}
\min\left(\frac{v(S^*(x,T))}{2},\frac{\alpha}{32}v(\opt(x,T))\right)\right)\rho(x)\mathrm{d}x.
\]

The optimal expected value is

\[
\int\limits_{\Omega}v(\opt(x))\rho(x)\mathrm{d}x.
\]

In the next we argue that function $g(x):=v(i^*(x))+\frac{1}{2^{|A|}}\sum_{T\subset A}
\min\left(\frac{v(S^*(x,T))}{2},\frac{\alpha}{32}v(\opt(x,T))\right)$ approximates $v(\opt(x))$
within a constant factor for any cost vector $x$.

We fix cost vector $x$. Let $v(i^*)=\frac{1}{k}v(\opt)$, for some $k\ge 1$. We know that due to
the Lemma~\ref{lem:probability} $\min\left(v(\opt(A\setminus T),v(\opt(T))\right)\ge\frac{k-1}{4k}v(\opt)$
at probability at least $\frac{1}{2}$. For each $T$ from the larger
``good'' half (i.e., for which previous inequality holds) we can write:

\[
v(\opt(T))\ge\frac{k-1}{4k}v(\opt)
\]

and

\begin{eqnarray*}
v(S^*)&\ge &v(S^*)-t\cdot x(S^*)\ge v(\opt(A\setminus T))-t\cdot x(\opt(A\setminus T))\\
&\ge& \frac{k-1}{4k}v(\opt)-t\cdot B\ge \frac{k-1}{4k}v(\opt)-\frac{v(\opt)}{8B}\cdot B\\
&\ge& \frac{k-2}{8k}v(\opt).
\end{eqnarray*}

The second inequality holds because $S^*\in\argmax\limits_{S\subseteq A\setminus T}\big\{v(S)-t\cdot x(S)\big\}$;
the third inequality holds because the cost of feasible solution $\opt(A\setminus T)$ is within the budget;
in the forth inequality we plugged in the definition of $t$ and used the fact that $v(\opt)\ge v(\opt(T))$.

Therefore, combining these two lower bounds for all ``good'' $T$ we get
\begin{eqnarray*}
g(x)& \ge &\left(\frac{1}{k}+\frac{1}{2}\min\left(\frac{k-2}{16k},\frac{(k-1)\alpha}{128k}\right)\right)v(\opt(x))\\
&\ge& \frac{1}{2}\left(\min\left(\frac{k+14}{16k},\frac{k\alpha}{128k}\right)\right)v(\opt(x))\\
&\ge&\frac{\alpha}{256}v(\opt(x)).
\end{eqnarray*}

Hence, we have shown that the expected value of \SubaddBayesian\ is within a constant factor of
$\frac{\alpha}{512}$ from the expected value of the optimal solution.

\section{Concentration Bounds for XOS functions}\label{ap:concentration}

In this subsection we discuss a deep, but, on the other hand, conceptually simple,
discovery in probability theory~\cite{BLM}, that has a direct application to our scenario.
We recommend a short survey by Vondrak~\cite{vondrak} on the various applications
of concentration inequalities for self-bounding functions in computer science.


Here we adduce one more example and possibly one more point of view on this fascinating subject.
Briefly in a simplified description the story is as follows. There is a set $A$ of $n$ elements
and each of them is picked independently at random into a set $T\subseteq A$. This set $T$
further receives a valuation $v(T)$ given by a positive monotone function $v:2^A\to\R$.
We note that when $v$ is an additive function, elements of $A$ naturally correspond to
coins of diverse face values, which then is being tossed and tails are taken into $T$.
When each coin has a small nominal relative to the total expected value, we get
exponential concentration around expectation by a natural variant of the low of large
numbers and Chernoff bounds (the standard formulation of LLN and consequently Chernoff bounds
is for 0-1 coins and different probabilities of tails, whereas in our case the probabilities are
all equal to $0.5$ but values may vary).

The most interesting part of the story is that one can generalize the exponential concentration
to much broader classes of valuations. For instance $v$ may be taken to be submodular or even
fractionally subadditive. Then the concentration inequalities are as follows.

\begin{claim}[\cite{vondrak}] Let $v(\cdot)$ be positive monotone XOS function.
If $\max_{i\in A}v(\{i\})= v_0,$ and $Z=v(X_1,\dots,X_{|A|})$,
where $X_i\in\{0,1\}$ are independently random, then for any $\delta>0.$
\begin{itemize}
\item $\pr\big[Z\ge(1+\delta)\Exp[Z]\big]\le\left(\frac{e^\delta}{(1+\delta)^{1+\delta}}\right)^{\frac{\Exp[Z]}{v_0}}.$
\item $\pr\big[Z\ge(1-\delta)\Exp[Z]\big]\le e^{-0.5\delta^2 \Exp[Z]/v_0}.$
\end{itemize}
\end{claim}

On the other hand, if one goes slightly further and considers subadditive valuations, then exponential
concentration fails to be true (see~\cite{vondrak} for counterexample). It is interesting to see one
more sharp distinction between subadditive and fractionally subadditive valuations in a quite different
mechanism design problem.